\documentclass{article}

\usepackage{arxiv}

\usepackage{hyperref}       
\usepackage{url}            
\usepackage{booktabs}       
\usepackage{graphicx}
\usepackage{natbib}
\usepackage{doi}

\usepackage[title]{appendix}%

\usepackage{amsmath}
\usepackage{amsthm}
\usepackage{amssymb}
\usepackage{enumerate}
\usepackage{color}
\usepackage{colortbl}
\usepackage{multirow}
\usepackage{float}
\usepackage{booktabs}
\usepackage{wrapfig}
\usepackage{subcaption}
\usepackage{authblk}

\newtheorem{theorem}{Theorem}[section]
\newtheorem{definition}[theorem]{Definition}

\newtheorem{lemma}[theorem]{Lemma}

\newcommand{\pkg}[1]{\textbf{#1}}

\title{Data-Driven Random Projection and Screening for High-Dimensional Generalized Linear Models}

\date{October 1, 2024}	

\author{\href{https://orcid.org/0000-0003-0893-3190}{\includegraphics[scale=0.06]{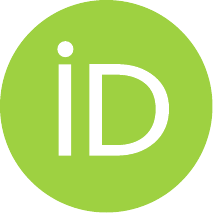}\hspace{1mm}Roman~Parzer} 
	\href{https://orcid.org/0000-0002-8014-4682}{\includegraphics[scale=0.06]{./plots/orcid.pdf}\hspace{1mm}Peter~Filzmoser} 
	\href{https://orcid.org/0000-0002-9613-7604}{\includegraphics[scale=0.06]{./plots/orcid.pdf}\hspace{1mm}Laura~Vana-G\"ur} \\
Institute of Statistics and Mathematical Methods in Economics\\
	TU Wien\\
	Vienna, Austria \\
	\texttt{romanparzer1@gmail.com} \\
}

\renewcommand{\shorttitle}{Random Projection and Screening for HD GLMs}

\hypersetup{
pdftitle={Data-Driven Random Projection and Screening for High-Dimensional Generalized Linear Models},
pdfsubject={stat.ME},
pdfauthor={Roman~Parzer, Peter~Filzmoser, Laura~Vana-G\"ur},
pdfkeywords={Generalized Linear Models, High-Dimensional Data, Predictive Modeling, Random Projection, Screening},
}

\begin{document}
\maketitle

\begin{abstract}
  We address the challenge of correlated predictors in high-dimensional GLMs, where regression coefficients range from sparse to dense, by proposing a data-driven random projection method.  This is particularly relevant for 
  applications where the number of predictors is (much) larger than the number of observations and the underlying structure -- whether sparse or dense -- is unknown.
  We achieve this by using ridge-type estimates for variable screening and random projection to incorporate information about the response-predictor relationship when performing dimensionality reduction.
  We demonstrate that a ridge estimator with a small penalty is effective for random projection and screening, but the penalty value must be carefully selected. Unlike in linear regression, where penalties approaching zero work well, this approach leads to overfitting in non-Gaussian families. Instead, we recommend a data-driven method for penalty selection.
  This data-driven random projection improves prediction performance over conventional random projections, even surpassing benchmarks like elastic net. Furthermore, an ensemble of multiple such random projections combined with probabilistic variable screening delivers the best aggregated results in prediction and variable ranking across varying sparsity levels in simulations at a rather low computational cost.
  Finally, three applications with count and binary responses demonstrate the method's advantages in interpretability and prediction accuracy.
\end{abstract}

\keywords{Generalized Linear Models \and High-Dimensional Data \and Predictive Modeling \and Random Projection \and Screening}

\section{Introduction}

High-dimensional data in a regression context, where the number of variables exceeds the number of observations (i.e., $p > n$ or even $p \gg n$), has become increasingly common across various applications, posing substantial computational and statistical challenges, 
particularly when dealing with discrete responses.
In such cases, predictors are often correlated and the sparsity of the true model is uncertain. Moreover, interpretability is increasingly becoming a model requirement in a variety of fields.
This calls for computationally efficient approaches that enable both accurate predictions and interpretable relationships between the predictors and the response.


The generalized linear model (GLM) extends the linear model to both continuous and discrete responses while maintaining interpretability. In high-dimensional settings, GLMs are typically regularized \citep[e.g.,][]{Tibsh1996LASSO,FanLi2001SCAD,Zou2005ElasticNet}.
Alternatively or complementarily, the dimensionality of the feature space can be reduced to a moderate size while learning and inference is performed in this reduced predictor space.  
One fast way to achieve this is \emph{variable screening}, i.e., selecting a subset of the predictors based on their utility. Methods for screening often rely on parametric \citep[such as the maximum likelihood estimates of univariate GLMs,][]{fan2009ultrahigh, Fan2010sisglms} or nonparametric \citep[e.g.,][]{fan2011nonparametric,mai2013kolmogorov,mai2015fusedkolmogorov,ke2023sufficient} measures but typically ignore predictor correlations. \cite{fan2009ultrahigh} suggest an iterative procedure to address this issue, while 
\cite{Wang2015HOLP} propose screening variables in linear regression using the high dimensional ordinary least squares projection (HOLP), a ridge-type estimator with a closed form solution when the penalty converges to zero. However, screening approaches based on ridge-type estimators are still rare in the context of GLMs.

An approach similar in scope to variable screening is random projection (RP), which reduces the dimensionality of the feature space by linearly projecting the features onto a lower dimensional space, rather than employing a reduced set of the original features. 
Conventional random projections contain iid entries from a suitable distribution and are oblivious to 
the data to be used in the regression. Such projections 
have been used in a classification setting in 
e.g.,~\cite{cannings2017random}, while \cite{BayComprRegr,Dunson2020TargRandProj} focus on linear regression.
On the other hand, \cite{ryder2019asymmetric} propose a data-informed random projection using an asymmetric transformation of the predictor matrix without using information of the response.
\cite{parzer2024sparse} propose a data-driven projection for linear regression which incorporates information from the estimated HOLP coefficients, i.e., about both the predictors and the response.

In this paper, we leverage the computational advantages of variable screening and random projection and introduce a data-driven random projection
method for GLMs that accounts for the relationship between
predictors and the response, while addressing the potentially complex correlation
structure among the predictors. For this purpose we propose a ridge-type estimator which can be  integrated into a sparse random projection matrix and can also be employed for screening the variables prior to projection, as it performs well in preserving the true
relationship between predictors and the response.

A key aspect of the proposed ridge-type estimator is the selection of the penalty term. We extend the HOLP estimator to  GLMs with canonical link, deriving a closed form solution, and explore if it retains the same benefits for random projection and screening in GLMs as in linear regression.
We find that for non-Gaussian families, a ridge estimator with zero penalty can overfit, making penalty selection a non-trivial task. While small penalty values reduce bias, a data-driven approach to choosing the penalty works best. Specifically, we propose selecting the smallest penalty value for which the deviance ratio
in the fit stays under a certain threshold (e.g., 0.8 for non-Gaussian families and 0.999 for Gaussian).
Simulations show that using these ridge estimates in the sparse RP matrix outperforms conventional RP techniques.

%
Given the randomness in the RP matrix, 
variability can be reduced by building ensembles of multiple RPs. For example, \citealt{BayComprRegr} propose repeatedly sampling RPs of different size and estimating an ensemble of linear regressions on the reduced predictors,  while 
\citealt{cannings2017random} generate various RPs for each classifier in an ensemble and pick the best one based on a appropriate loss function.
Additionally, ensembles of multiple RPs and variable screening steps have also been proposed 
achieve better 
predictive performance in linear regression \citep{Dunson2020TargRandProj,parzer2024sparse}.

In a similar fashion we extend the variable screening and random projection procedure by building an ensemble of GLMs and averaging them to form the final model,  adapting the algorithm in \cite{parzer2024sparse} to GLMs.
An extensive simulation study reveals that this ensemble improves predictive performance, ranks predictors effectively, and is computationally efficient, particularly with increasing predictor dimensions. It consistently performs well against state-of-the-art approaches such as penalized regression, random forests, and support vector machines (SVMs) across a range of sparsity settings and 
yields the best overall performance when aggregated across all scenarios. 
This broad applicability makes the method versatile for
high-dimensional regression with correlated predictors, especially when the sparsity of the underlying data generating process is unknown, and capable to computationally handle datasets with small $n$ and a (very) large number of predictors. 

The integration of the GLM framework with probabilistic screening improves model interpretability, as 
the coefficients in the marginal models can be extracted and their reliability and relevance can be assessed. The GLM framework also offers modeling flexibility, facilitating seamless comparison across  different family-link combinations.

The paper is organized as follows: 
Section~\ref{sec:method} introduces the GLM model, the proposed data-informed random projection and variable screening in GLMs
as well as the ensemble algorithm.
In Section~\ref{sec:sims}, extensive simulations motivating and comparing the method with state-of-the-art 
approaches are performed.
Applications are presented in 
Section~\ref{sec:data}.
Section~\ref{sec:conclusion} concludes the paper.

\section{Method}\label{sec:method}

This section begins by presenting the GLM model class, followed by an introduction to
the dimension reduction tools random projection and variable screening. Then, we propose a novel coefficient estimator useful for both these concepts, and state, how this estimator can be used to extend the algorithm in \cite{parzer2024sparse} to GLMs.
Throughout the section, we use the notation $[n]=\{1,\dots,n\}$ for any $n\in\mathbb{N}$.

\subsection{Generalized linear models}\label{sec:model}

We assume to observe high-dimensional data $\{(\boldsymbol{x}_i,y_i)\}_{i=1}^n, \boldsymbol{x}_i\in\mathbb{R}^p,y_i\in\mathbb{R}$ with $p\gg n$ from a 
GLM with the responses having conditional densities from a (reproductive) 
exponential dispersion family of the form
\begin{align}\label{eqn:y_density}
  f(y_i|\theta_i,\phi) = \exp\Bigl\{\frac{y_i\theta_i- b(\theta_i)}{a(\phi)} + c(y_i,\phi)\Bigr\},
\end{align}
where $\theta_i$ is the natural parameter,  $a(.)>0$ and $c(.)$ are specific real-valued functions determining different families, $\phi$ is a dispersion parameter, and 
$b(.)$ is the log-partition function normalizing the density to integrate to one. If $\phi$ is known, we obtain densities in the natural exponential family for our responses. 
It can be shown that $b(.)$ is twice differentiable and convex with $\mathbb{E}[y_i|\theta_i,\phi] = b'(\theta_i)$ and $\text{Var}(y_i|\theta_i,\phi) = a(\phi) b''(\theta_i)>0$, if the responses have (positive) second moments \citep[see e.g.][Section 2.2.2]{mccullagh1989GLM}.

The responses are related to the $p$-dimensional predictors through the conditional mean, i.e.,~the conditional mean of $y_i$ given $\boldsymbol{x}_i$
depends on a linear combination of the predictors through a (invertible) link function $g(.)$
\begin{align}\label{eqn:link}
  g(\mathbb{E}[y_i|\boldsymbol{x}_i]) = \beta_0 + \boldsymbol{x}_i'\boldsymbol{\beta}=:\eta_i,
\end{align}
where $\beta_0\in\mathbb{R}$ is the intercept and $\boldsymbol{\beta}\in\mathbb{R}^p$ is a vector of regression coefficients.
Equations~\eqref{eqn:y_density}~and~\eqref{eqn:link} give the functional relation $\theta_i = \theta_i(\beta_0,\boldsymbol{\beta},\boldsymbol{x}_i) = (b')^{-1}(g^{-1}(\eta_i))$
between $\theta_i$ and $\eta_i$. For each family, $g:=(b')^{-1}$ is the canonical link function, such that $\theta_i=\eta_i$. 
The full log-likelihood of the regression parameter $\boldsymbol{\beta}$ given the data  $\{(\boldsymbol{x}_i,y_i)\}_{i=1}^n$ is
\begin{align*}
  \ell(\beta_0,\boldsymbol{\beta}) = \sum_{i=1}^n \frac{y_i\theta_i(\beta_0,\boldsymbol{\beta},\boldsymbol{x}_i)- b(\theta_i(\beta_0,\boldsymbol{\beta},\boldsymbol{x}_i))}{a(\phi)} + c(y_i,\phi),
\end{align*}
but for maximization with respect to $\boldsymbol{\beta}$, it suffices to use 
\begin{align*}
  \tilde \ell(\beta_0,\boldsymbol{\beta}) = \sum_{i=1}^n y_i\theta_i(\beta_0,\boldsymbol{\beta},\boldsymbol{x}_i)- b(\theta_i(\beta_0,\boldsymbol{\beta},\boldsymbol{x}_i)),
\end{align*}
and treat $\phi$ as a nuisance parameter. 
In our general high-dimensional setting $p>n$, the predictor matrix $\boldsymbol{X}\in\mathbb{R}^{n\times p}$ with rows $\boldsymbol{x}_i$
can be assumed to have full $\text{rank}(\boldsymbol{X})=n$, so there is typically a whole (affine) subspace of $\boldsymbol\beta$s yielding the same $\eta_i$s
and we can not hope to find a unique maximizer $\boldsymbol{\hat\beta}$. 
In order to reduce the dimension of the problem, we resort to two techniques namely, random projection and variable screening.

\subsection{Random projection and variable screening }\label{sec:method:RP}

\paragraph{Random projection} can be used as a dimension-reduction tool for high-dimensional regression
by creating a random matrix $\Phi\in\mathbb{R}^{m\times p}$ with $m\ll p$ and using the reduced predictors
$\boldsymbol{z}_i=\Phi \boldsymbol{x}_i \in\mathbb{R}^m$ in a regression model. 
When using this method for
GLMs, we would like the predictors to still have most of the 
predictive power and that the true regression coefficients $\boldsymbol{\beta}$ are close to the row span of $\Phi$, 
such that they can be approximately recovered by the reduced predictors.
For this purpose, we propose to employ the following random projection. 
\begin{definition}\label{def:RP_CW}
  Let $h:[p]\to[m]$ be a random map such that for each $j\in[p]: h(j)=h_j \overset{iid}{\sim} \text{Unif}([m])$.
  Let $\boldsymbol{B}\in\mathbb{R}^{m\times p}$ be a binary matrix with ${B}_{h_j,j}=1$ for all $j\in[p]$ and remaining entries $0$, where we assume $\text{rank}(\boldsymbol{B})=m$.
  Let $\boldsymbol{D}\in \mathbb{R}^{p\times p}$ be a diagonal matrix with entries $d_j \in\mathbb{R}\setminus\{0\}, j\in[p]$.
  Then we call $\boldsymbol{\Phi} = \boldsymbol{BD}$ a CW random projection. 
\end{definition}
When using random sign diagonal elements $d_j \sim \text{Unif}(\{-1,1\})$ independent of $h$, we obtain a
\textit{sparse embedding matrix} $\boldsymbol{\Phi}\in\mathbb{R}^{m\times p}$, $m\ll p$ from \citet{Clarkson2013LowRankApprox}. 
Aside from being sparse and computationally efficient, 
this random projection also exhibits the property
$\boldsymbol{d}=(d_1,\dots,d_p)'\in\text{span}(\boldsymbol{\Phi'})$. Thus, by choosing $d_j\propto \beta_j, j=1,\dots,p$ instead of random sign diagonal elements, we can reach our goal of combining the variables to reduced predictors with strong predictive power which are able to recover the true regression coefficient.
In Theorem 1 of \cite{parzer2024sparse}, it is shown that this approach significantly reduces the expected squared error of future predictions in the linear regression setting.  
Below, we will propose a new estimator for a general family to use as diagonal elements.

\paragraph{Variable screening} aims at selecting a small subset of variables based on some marginal utility measure, 
and using the ones with the highest utility for further analysis. 
This procedure can complement the RP approach and further reduce the dimensionality of the problem, by first screening for the important variables and then performing the random projection. 

A seminal contribution in variable screening is the sure independence screening (SIS) of \cite{Fan2007SISforUHD}, who proposed to use the absolute marginal correlation of the predictors to the response in linear regression, which was later extended to GLMs in \cite{fan2009ultrahigh, Fan2010sisglms} by employing the maximum likelihood coefficient estimates of univariate GLMs instead of the correlation coefficient. 
However, in the presence of predictor correlation, screening based on a conditional utility measure (i.e., conditional on all the other variables in the model) is to be preferred to the unconditional one. 
To tackle this issue, \cite{Fan2007SISforUHD} propose iterative SIS,  which involves iteratively applying SIS and penalized regression to select a small set of variables, computing the residuals of a model fitted with the selected variables and using these residuals residuals as a response variable to continue finding relevant variables. This procedure was later extended to more general model classes in \cite{fan2009ultrahigh}.

In general, in a GLM each variable's utility given all the other variables amounts to $|\beta_j|$, so 
another approach is to find a screening coefficient which capable of detecting the correct order of magnitudes of the regression coefficients (but not necessarily their signs). We note that an estimator
which performs well for the purpose of random
projection (as discussed above) would also be a good candidate as a screening coefficient.
In the next section, we propose such an estimator for GLMs.

\subsection{Proposed estimator for random projection and screening}
In general, a ridge-type estimator 
\begin{align}\label{eqn:beta_lam}
    \boldsymbol{\hat{\beta}}_{\lambda}  = \text{argmin}_{\boldsymbol{\beta}\in\mathbb{R}^p}\min_{\beta_0\in\mathbb{R}} \Bigl\{-\tilde{\ell}( \beta_0,\boldsymbol{\beta}) + \frac{\lambda}{2}\sum_{j=1}^p{{\beta}}_j^2\Bigr\}, \quad \lambda>0
\end{align}
promises to be a sensible choice, both for screening and for inclusion in the CW random projection, because it considers all variables in the model and is non-sparse.
For linear regression, the HOLP estimator \citep{Wang2015HOLP} provides a screening utility measure considering all variables' effects simultaneously and has strong theoretical properties (see Section 2.1 in \cite{parzer2024sparse} for a discussion). It is explicitly given by
\begin{align}\label{eqn:HOLP}
  \boldsymbol{\hat\beta}_{\text{HOLP}} = \boldsymbol{X}'(\boldsymbol{X}\boldsymbol{X}')^{-1} \boldsymbol{y} = \lim_{\lambda \to 0} \Bigl( \text{argmin}_{\boldsymbol{\beta}\in\mathbb{R}^p}\Bigl\{ \sum_{i=1}^n (y_i  - \boldsymbol{x}_i'\boldsymbol{\beta})^2 + \frac{\lambda}{2} \sum_{j=1}^p{\beta}_j^2\Bigr\}\Bigr) ,
\end{align}
where $\boldsymbol{X}\in\mathbb{R}^{n\times p}$ is the predictor matrix and $\boldsymbol{y}\in\mathbb{R}^n$ is the response vector. Here, a model without an intercept  $\beta_0$ is assumed, which can be justified by centering $\boldsymbol{X}$ and $\boldsymbol{y}$. \cite{parzer2024sparse} used this HOLP coefficient \eqref{eqn:HOLP} as the diagonal elements in a CW random projection.

Motivated by \eqref{eqn:HOLP} and \cite{Kobak2020RidgeinHD}, who showed that the optimal ridge-penalty in linear regression can be negative due to implicit regularization from high-dimensional predictors, 
we first investigate whether $\lim_{\lambda\to 0}  \boldsymbol{\hat{\beta}}_{\lambda}$ is also an appropriate choice in GLMs. 
For logistic regression (binomial family with logit link), it is known 
that the estimator~\eqref{eqn:beta_lam}, scaled by its norm, converges to a 
hard-margin support vector machine (SVM) coefficient for $\lambda\to 0$ (Theorem 3 in \cite{Rosset2004BoostingMaxMar}, see Section 3.6.1 in \cite{Hastie2015STwSparsityBook} for a discussion). More generally, the following Theorem shows an explicit form of $\lim_{\lambda\to 0}  \boldsymbol{\hat{\beta}}_{\lambda}$ for families with canonical link function such that $g(y_i)\in\mathbb{R}$ for all $i\in[n]$.

\begin{theorem}\label{theorem:lim_beta_lam}
  For a family with canonical link function satisfying $g(y_i)\in\mathbb{R}$ for all $i\in[n]$, for a full rank predictor matrix $\text{rank}(\boldsymbol{X}\boldsymbol{X}')=n$, and in a model without an intercept $\beta_0$, we obtain
  \begin{align}\label{eq:holpglm}
    \lim_{\lambda\to 0}  \boldsymbol{\hat{\beta}}_{\lambda} = \boldsymbol{X}'(\boldsymbol{X}\boldsymbol{X}')^{-1}g(\boldsymbol{y}).
  \end{align}
\end{theorem}
The proof can be found in the Appendix~\ref{sec:app:proof}. The intercept-free assumption can be avoided by appropriate centering of $\boldsymbol{X}$ and $g(\boldsymbol{y})$. 
For practical usage, this exact limit has a few drawbacks. For example, when centering by the exact sample mean, $\boldsymbol{X}$ would not have full rank $n$ anymore. A workaround could be employing a generalized inverse or using different location estimators (like median or trimmed mean) for the purpose of centering the variables. In the cases where $g(y_i)\notin\mathbb{R}$, one could approximate $g(y_i)$ by using a continuity correction, but there
is no guarantee that this approximation works well. Also, Theorem~\ref{theorem:lim_beta_lam} only covers the canonical link for each family. 

As an alternative that can cover all cases, we propose to approximate  $\lim_{\lambda\to 0}  \boldsymbol{\hat{\beta}}_{\lambda}$
by a ridge estimator with a small fixed $\lambda_{\min}>0$ and to use $\boldsymbol{\hat\beta}_{\lambda_{\min}}$ from \eqref{eqn:beta_lam} for variable screening and as the diagonal elements in the 
data-informed random projection.

Furthermore, in order to understand whether more
penalization (i.e.,~through higher values of $\lambda$) is needed for other families as compared to the Gaussian case, we investigate 
alternative strategies to choosing the penalty value.
In a simulation example in Section~\ref{sec:sims:RP}, we investigate the choice of this $\lambda_\text{min}$ for different families. It shows, that for non-Gaussian families, it is beneficial to use a higher $\lambda_\text{min}$ to avoid a saturated fit. We also show that the resulting estimator allows good recovery of sign and magnitude of the true non-zero coefficients, 
while also performing well in terms of prediction when employed in the RP matrix. 

\subsection{SPAR algorithm for GLMs}\label{sec:method:alg}

Employing one single data-driven random projection with the proposed estimator and then estimating a GLM on the reduced predictors can lead to 
high variability due to randomness.
We address this by adapting the Sparse Projected Averaged Regression (SPAR) algorithm from \cite{parzer2024sparse} to GLMs, which builds an ensemble of GLMs in the following way: i) randomly sampling predictors for inclusion in the random projection based on the proposed screening coefficient, ii) projecting the sampled variables to a randomly chosen lower dimension using the proposed random projection, iii) estimating penalized GLMs with the reduced predictors and iv) averaging them to form the final model.
The adapted algorithm is given below, where $*$ mark changes compared to the linear regression formulation:
\begin{enumerate}
  \setlength\itemsep{0em}
    \item[$1.^*$] choose family with corresponding log-likelihood $\ell(.)$ and link and standardize covariate inputs $\boldsymbol{X}:n\times p$
    \item[$2.^*$] calculate $\boldsymbol{\hat\alpha}=\boldsymbol{\hat\beta}_{\lambda_{\min}}=\text{argmin}_{{\boldsymbol{\beta}}\in\mathbb{R}^p}\min_{\beta_0\in\mathbb{R}}\Bigl\{ -\tilde \ell( \beta_0,\boldsymbol{\beta}) + \frac{\lambda_\text{min}}{2}\sum_{j=1}^p{{\beta}}_j^2\Bigr\}$, see below for choice of $\lambda_\text{min}>0$ 
    \item[3.] For $k=1,\dots,M$:
     \begin{enumerate}
      \item[3.1.] draw $2n$ predictors with probabilities $p_j\propto |{\hat\alpha}_{j}|$ yielding screening index set $\boldsymbol{I}_k=\{j_1^k,\dots,j_{2n}^k\}\subset[p]$
      \item[3.2.] project remaining variables to dimension $m_k\sim \text{Unif}\{\log(p),\dots,n/2\}$ using $\boldsymbol{\Phi}_k:m_k\times 2n$ from Definition~\ref{def:RP_CW} with diagonal elements 
      $d_i={\hat \alpha}_{j_i^k}$
      to obtain reduced predictors $\boldsymbol{Z}_k=\boldsymbol{X}_{\cdot I_k}\boldsymbol{\Phi}_k' \in \mathbb{R}^{n\times m_k}$
      \item[$3.3.^*$] fit a GLM of $\boldsymbol{y}$ against $\boldsymbol{Z}_k$ (with small $L_2$-penalty) to obtain estimated coefficients $\boldsymbol{\gamma}^k\in\mathbb{R}^{m_k}$ and $\boldsymbol{\hat\beta}^k$, where $\boldsymbol{\hat\beta}_{I_k}^k=\boldsymbol{\Phi}_k'\boldsymbol{\gamma}^k$ and $\boldsymbol{\hat\beta}_{\bar I_k}^k=0$.
    \end{enumerate}
    \item[4.] for a given threshold $\nu>0$, set all entries ${\hat\beta}_j^k$ with $|{\hat\beta}_j^k|<\nu$ to $0$ for all $j,k$
    \item[5.] combine via simple average {on link-level} $\boldsymbol{\hat\beta} = \sum_{k=1}^M\boldsymbol{\hat\beta}^k / M$ {or on response level $\boldsymbol{\hat y} = \sum_{k=1}^M\boldsymbol{\boldsymbol{\hat y}}^k / M$}
    \item[$6.^*$] \textit{optional:} choose $M$ and $\nu$ via $10$-fold cross-validation by repeating steps 2 to 6
    (with fixed index sets $\boldsymbol{I}_k$ and projections $\boldsymbol{\Phi}_k$) for each fold and 
    evaluating the performance by model deviance (Dev) on the withheld fold; 
    and choose  
    \begin{align*}
      (M_{\text{best}},\nu_{\text{best}}) = \text{argmin}_{M,\nu}\text{Dev}(M,\nu)
    \end{align*}

    \item[7.] output the estimated coefficients and predictions for the chosen $M$ and $\nu$
  \end{enumerate}

We propose the following strategy for choosing the penalty $\lambda_\text{min}$. Along a decreasing path of $\lambda$s (e.g.,
equally spaced on the logarithmic scale), the fitted model will get closer and closer to a saturated fit and a deviance ratio of one. We propose to use the smallest $\lambda$, where the deviance ratio does not exceed to a certain threshold. It turns out that this threshold can be set to a value close to one (e.g.,~$0.999$) for the Gaussian family but to a lower value (e.g.,~$0.8$) for other families, as we will show in subsequent simulations. 

In fitting the marginal models in step 3.3, we also obtain intercept estimates, which can also be averaged and translated back to the original predictor scale to give an overall estimate $\hat\beta_0$ of the intercept $\beta_0$.

This algorithm allows for several variations. For instance, different measures can be used in the cross-validation, such as mean squared error instead of the family-dependent deviance.
If stricter thresholding is desired, $(M,\nu)$ can be chosen by the one-standard-error rule, which yields the sparsest $\boldsymbol{\hat\beta}$ within one standard error of the score of the best parameters.
While averaging in step 6 can also be done on response-level, simulations showed no performance gain in any setting, and averaging at the linear predictor-level is more interpretable, as a single final coefficient estimate $\boldsymbol{\hat\beta}$, as well as a distribution of each $\hat\beta_j$ over the marginal models can be reported.

\section{Simulation study}\label{sec:sims}

In a first simulation study we compare how well the estimator \eqref{eqn:beta_lam}  recovers 
the true active $\boldsymbol\beta$ across different penalty choices, and evaluate the predictive performance of the data-driven random projection with these estimators. We then compare 
SPAR algorithm's predictive performance and variable ranking ability against various benchmarks in a comprehensive simulation study.

\subsection{Setup}
We generate data from Equation~\eqref{eqn:y_density} for five family-link combinations with $n=200$ and use additionally $n_\text{test}=200$ observations as a test sample.  
The $p$-dimensional predictors are simulated as $\boldsymbol{x}_i\sim N_p(\boldsymbol 0,\boldsymbol{\Sigma})$, where we investigate different structures for $\boldsymbol{\Sigma}$: \emph{identity}, 
\emph{compound}
with  ${\Sigma}_{ij}=0.5$ if $i\neq j$ and $1$ otherwise, \emph{autocorrelated} with ${\Sigma}_{ij}=0.9^{|i-j|}$, and a \emph{block}
structure (blocks of size 100 where half of the blocks have a compound, the remaining blocks except the last have an autocorrelated structure, plus independent predictors in the last block.

We consider $p=500, 2000, 10000$ and three sparsity settings for $\boldsymbol{\beta}$: sparse ($a=[2\log(p)]$), medium ($a=[2\log(p)+n/2]$) and dense ($a=p/4$),
where $a$ is the number of non-zero entries in
$\boldsymbol{\beta}$. 
These entries are independently set as $(-1)^u (4\log(n)/\sqrt{n} +|z|)$ at uniformly random positions, where $u\sim \text{Bernoulli}(0.4)$ and $z\sim N(0,1)$ \citep{Fan2007SISforUHD}.
Finally, for each family-link we rescale $\boldsymbol{\beta}$ to control the signal strength by $\boldsymbol{\beta}'\boldsymbol{\Sigma}\boldsymbol{\beta}=c$, where $c=100,1000,0.25,10,0.125$ for binomial-logit, binomial-complementary log log (cloglog), Poisson-log, Gaussian-identity and Gaussian-log, respectively.
The intercept $\beta_0$ is set such that  $\sum_{i=1}^n\mathbb{E}[y_i|\boldsymbol{x}_i]/n=0.5,0.7,10,1,10$ for 
the respective family-links. These values are chosen to ensure that the problem is not too easy to solve but also that most methods can explain some of the deviance.

\subsection{Measures}
Prediction performance is assessed on the independent test samples 
$\{(\boldsymbol{x}_{n+i},y_{n+i}):i\in[n_\text{test}]\}$ by several measures:
\emph{mean squared prediction error} and its relative version
$$
\text{MSPE}(\boldsymbol{\hat y})= \frac{1}{n_\text{test}}\sum_{i=1}^{n_\text{test}}(y_{n+i} - \hat y_{n+i})^2,\quad
\text{rMSPE}(\boldsymbol{\hat y})=n_\text{test}\text{MSPE}(\boldsymbol{\hat y}) / \sum_{i=1}^{n_\text{test}}(y_{n+i} - \bar y)^2
$$
where $\bar y$ is the mean of the responses in the training sample (for the binomial family, MSPE corresponds to the Brier score). We also compute 
\emph{mean squared link estimation error} to assess the linear predictor
accuracy
$$
\text{MSLE}(\boldsymbol{\hat\beta})=\frac{1}
{n_\text{test}}\sum_{i=1}^{n_\text{test}}(\eta_i - \hat\eta_i)^2,\quad
\text{rMSLE}(\boldsymbol{\hat\beta})=n_\text{test}\frac{\text{MSLE}(\boldsymbol{\hat\beta})}{\sum_{i=1}^{n_\text{test}}(\hat \eta_i)^2},
$$
where $\eta_i =\beta_0 + \boldsymbol x_{n+i}'\boldsymbol{\beta}$ and $\hat\eta_i=\hat\beta_0 + \boldsymbol x_{n+i}'\boldsymbol{\hat\beta}$.
For the binomial family we also consider the area under the receiver operating characteristic curve (AUC) of the predicted probabilities 
$\hat p_{n+i} = 1/(1+\exp(-\hat\eta_i))$ to the binary responses.
 
Furthermore, we assess variable ranking using the \emph{partial AUC} (pAUC) which considers whether a variable is truly active and the absolute value of their estimated coefficient. To allow fair comparison between sparse methods and  methods delivering a dense coefficient vector, we limit false positive to $n/2$ \citep[see also][]{parzer2024sparse}.
\subsection{Simulations for screening and random projection}\label{sec:sims:RP}

\paragraph{Recovery of true active $\boldsymbol{\beta}$}
We investigate whether the ridge estimator in \eqref{eqn:beta_lam}
is appropriate for the purpose of screening and data-informed random projection in the sense that it is able to recover the true active coefficients well by considering the correlation to the true active  $\boldsymbol{\beta}$.
We present results for $n=200, p=2000$, medium sparsity and $\boldsymbol{\Sigma}$  block-diagonal.
Figure~\ref{fig:Scr} shows the distribution of the correlation coefficient of the true active coefficients to different screening coefficients over 100 replications. We consider the ridge (L2) estimator from~\eqref{eqn:beta_lam} with the following choices for $\lambda$:  chosen by cross-validation based on the deviance criterion (L2$\_$cv); fixed at 10 (L2$\_$10), 1 (L2$\_$1), 0.1 (L2$\_$01), 0.01 (L2$\_$001);
the estimates for the penalty converging to zero (L2$\_$limit0). For binomial-logit we use the exact limit for $\lambda\to 0$, which is a hard-margin linear support-vector-machine (SVM) coefficient (estimated with \textsf{R} package \pkg{e1071} with cost set to $10^{10}$; \citealt{e1071R}). For the other families we use Equation~\ref{eq:holpglm} and impute any zeros present in the response variable by a small positive value for the poisson family. 
Furthermore, we employ a data-driven approach to choosing $\lambda$ as the smallest value  for which the fraction of (null) deviance explained does not exceed 80\%  (L2$\_$dev08), 95\% (L2$\_$dev095) and 99.9\%  (L2$\_$dev0999).
Finally, we examine the screening coefficient in \cite{Fan2010sisglms}, where each ${\beta}_j$ is estimated by the slope of a marginal GLM of the corresponding predictor
 $\{{x}_{ij},i\in[n]\}$ with an intercept (marGLM).

We observe that L2$\_$limit0 does not deliver the best results for all investigated family links (we omit the results for binomial-cloglog as they are similar to binomial-logit). 
L2$\_$cv also underperforms, and the marginal GLM coefficients are the least effective at recovering the true coefficients.
Fixing $\lambda$ to predetermined values is
also not appropriate for all family links, since the strength of the penalty varies between families. Table~\ref{tab:scr_lambda} in Appendix shows the average resulting $\lambda$ for each family using cross-validation and the different deviance cut-offs.
For binomial, differences among the ridge estimators are minor, but for  Poisson the performance declines as the deviance ratio threshold increases or $\lambda$ becomes too value. While the cross-validated penalty seems to be too large for the purpose of screening,  
using a 0.8 deviance ratio for non-Gaussian families and 0.999 for  Gaussian responses yields good results. The corresponding estimates are therefore 
well-suited as diagonal elements approximately proportional to the true $\boldsymbol{\beta}$ in our data-informed random projection from Definition~\ref{def:RP_CW}.
In this comparison, we also computed the pAUC of these coefficients, and the ratio of true active variables within the highest $3a$ absolute estimated values, but omit these results as they are similar to those based on correlation.
\begin{figure}
  \centering
  \includegraphics[width=1.0\columnwidth, trim=0 20 0 5,clip]{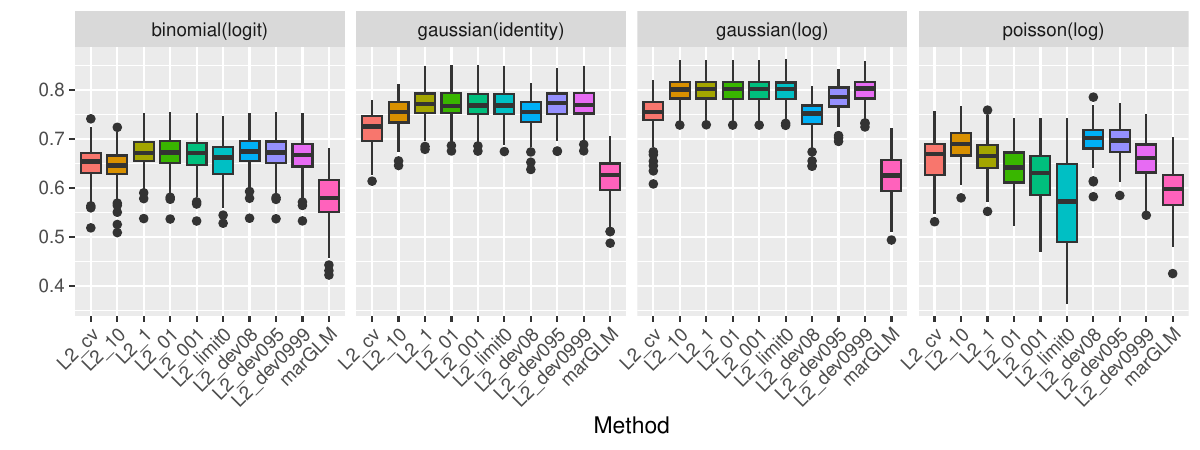}
  \caption{Correlation of true active coefficients to different screening estimators for $100$ replications 
  ($n=200,p=2000$, medium sparsity and 
  block-diagonal $\boldsymbol{\Sigma})$.}
  \label{fig:Scr}
\end{figure}
\paragraph{Data-informed random projection}
Next, we investigate the predictive performance of a model where the predictors are first projected onto a $m=n/4=50$ lower dimensional space using the proposed data-informed random projection with different screening coefficients, namely L2\_cv, L2\_limit0, L2\_dev08, L2\_dev095 and L2\_dev0999 from the previous section. Additionally, we also show the oracle performance of our proposal using the true $\boldsymbol{\beta}$ in the random projection (True\_Beta). 
We also consider models where conventional RPs (i.e.,~Gaussian with $iid$ standard normal entries and SparseCW from Definition~\ref{def:RP_CW} with random sign diagonal elements) are used.
Furthermore, we estimate the models with adaptive LASSO \citep{Zou2006AdLASSO} and elastic net ($\alpha=3/4$) with penalty chosen by CV based on deviance as performance benchmarks.

In Figure~\ref{fig:RP} we present the rMSLE and the prediction error ($1-$AUC for binomial, rMSPE for Gaussian and Poisson) for $n_\text{test}=200$ over 100 repetitions and see that 
the proposed data-informed projection with ridge estimates in the diagonal generally increases the performance with respect to both metrics over the conventional RPs (Gaussian and SparseCW), reaching a lower link estimation error and a prediction error. The differences between the ridge estimators are less obvious, but generally, we see that the performance of the estimators in terms of screening also translates to the prediction power, namely that L2\_dev0999 or L2\_limit0 deliver the best results for the Gaussian family with both identity and log link, while L2\_dev08  achieves the best prediction performance for the other families.

For this example, it can be seen that this adaption of the diagonal of the projection matrix suffices to reach a better performance than high-dimensional regression benchmarks adaptive LASSO and elastic net, but there is a noticeable gap to the oracle performance.

\begin{figure}
  \centering
  \includegraphics[width=1.0\columnwidth, trim=0 20 0 5,clip]{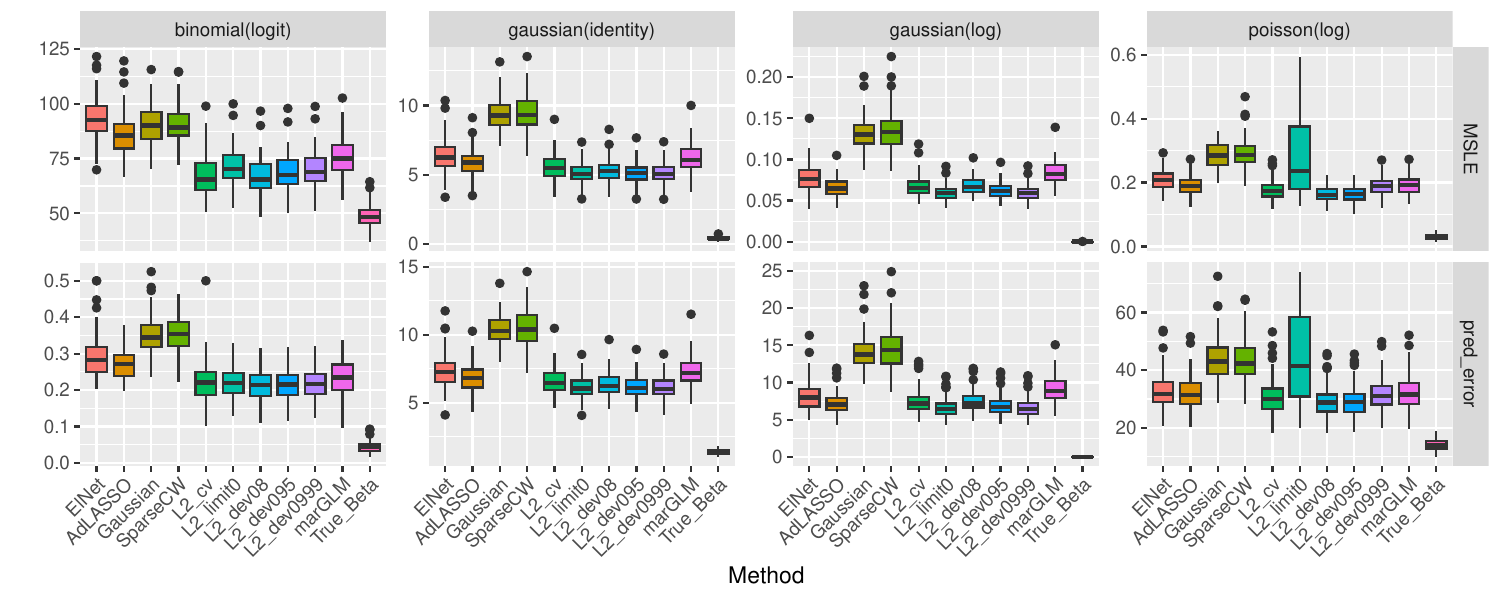}
  \caption{Link estimation error and prediction error ($1-$AUC for Binomial, MSPE otherwise) of Adaptive LASSO, Elastic Net, Gaussian and CW random projections, and our proposed projection with different ridge estimators, 
  for $100$ replications ($n=200, p=2000$, medium sparsity and block-diagonal $\boldsymbol{\Sigma}$).}
  \label{fig:RP}
\end{figure}

\subsection{Benchmark simulations for SPAR}\label{sec:sims:SPAR}

We consider $n=200,p=2000$ in sparse, medium and dense setting for the five family-link combinations and a block structure for $\boldsymbol{\Sigma}$. We report here results for the canonical links.
In the supplementary materials, we report additional results, which include different covariance structures for $p=2000$ predictors for the medium setting, as well as different values of $p$ for the medium setting with block-covariance.

We compare the following methods. 
Assuming the GLM, we use LASSO, elastic net ($\alpha=3/4$), ridge, adaptive LASSO (using package \pkg{glmnet} \citealt{glmnet2023})
and sure independence screening \citep[SIS;][]{Fan2010sisglms}. 
Then, as general regression benchmarks, we use random forest (RF implemented \textsf{R}-package \pkg{randomForest} with \texttt{mtry} parameter tuned by CV; \citealt{randomForestR}) and
support vector machines (\textsf{R} package \pkg{e1071} with cost and kernel -- linear or radial -- tuned by CV; \citealt{e1071R}). For these two methods, we do not report results for link estimation, since they do not (necessarily) estimate a linear predictor. For variable ranking, we use the reported importance measure (i.e.,~(weighted) mean of the individual trees' decrease in Gini-index or MSE produced by each variable) for RF and
the inner product of the estimated coefficients and the support vectors for SVM.

Finally, as a set of methods using random projections, we use an ensemble of $50$ models with the conventional RP from Definition~\ref{def:RP_CW} with random sign entries without any screening and random goal dimensions as in step 3.2 of the SPAR-algorithm in Section~\ref{sec:method:alg} (RP\_CW\_Ensemble), Targeted Random Projections (TARP), which is an adaptation of \cite{Dunson2020TargRandProj} to GLMs where we perform screening based on marginal 
GLM coefficients and use the conventional RP of \cite{ACHLIOPTAS2003JL} with $\psi=1/6$ for an ensemble of $20$ models, 
and our proposed SPAR algorithm from Section~\ref{sec:method:alg}, once without cross-validation with fixed $20$ models and $\lambda$ chosen from a grid of values based on the 
best model deviance achieved on the training set, and once with cross-validation for a grid of $\lambda$ values and up to $50$ models.

Figure~\ref{fig:pred} shows prediction and link estimation performance for $100$ replications of $n=200,p=2000$ in the block covariance setting. The results for LASSO and SIS are omitted in this figure to provide a more comprehensive overview since they are always outperformed by AdLASSO and ElNet, respectively.
Generally, SPAR and SPAR-CV are outperformed by AdLASSO and elastic net in terms of prediction and link estimation only in the sparse settings, while they are among the best performing methods in all other settings. Especially the performance in link estimation for the logistic regression is remarkable. Figure~\ref{fig:sim_pred_cov_settings} in the Supplementary materials shows the prediction error for the five family-link combinations across the different covariance settings for $\boldsymbol{\Sigma}$ for fixed medium sparsity. The results for autocorrelated and block covariance are rather similar, while all methods perform better for the compound covariance and worse for the identity covariance, which could be due to the amount of information contained in the covariance structure. As final inspection of prediction performance, Figure~\ref{fig:sim_pred_p} in the Appendix shows prediction errors for increasing $p$ in the block covariance and medium sparsity. All methods lose some performance when the dimension increases, but adaptive LASSO and elastic net seem to suffer the most. 
Figure~\ref{fig:pAUC} shows that SPAR and SPAR-CV both perform well in terms of variable ranking as measured by pAUC, where only SVM yields a similar performance across all settings. For the ensemble methods, we here compute the pAUC of the final averaged $\boldsymbol{\hat\beta}$.
As expected, AdLASSO and ElNet perform best in the sparse settings. 

\begin{figure}[t!]
\centering
\includegraphics[width=0.9\columnwidth,trim=0 20 0 6,clip]{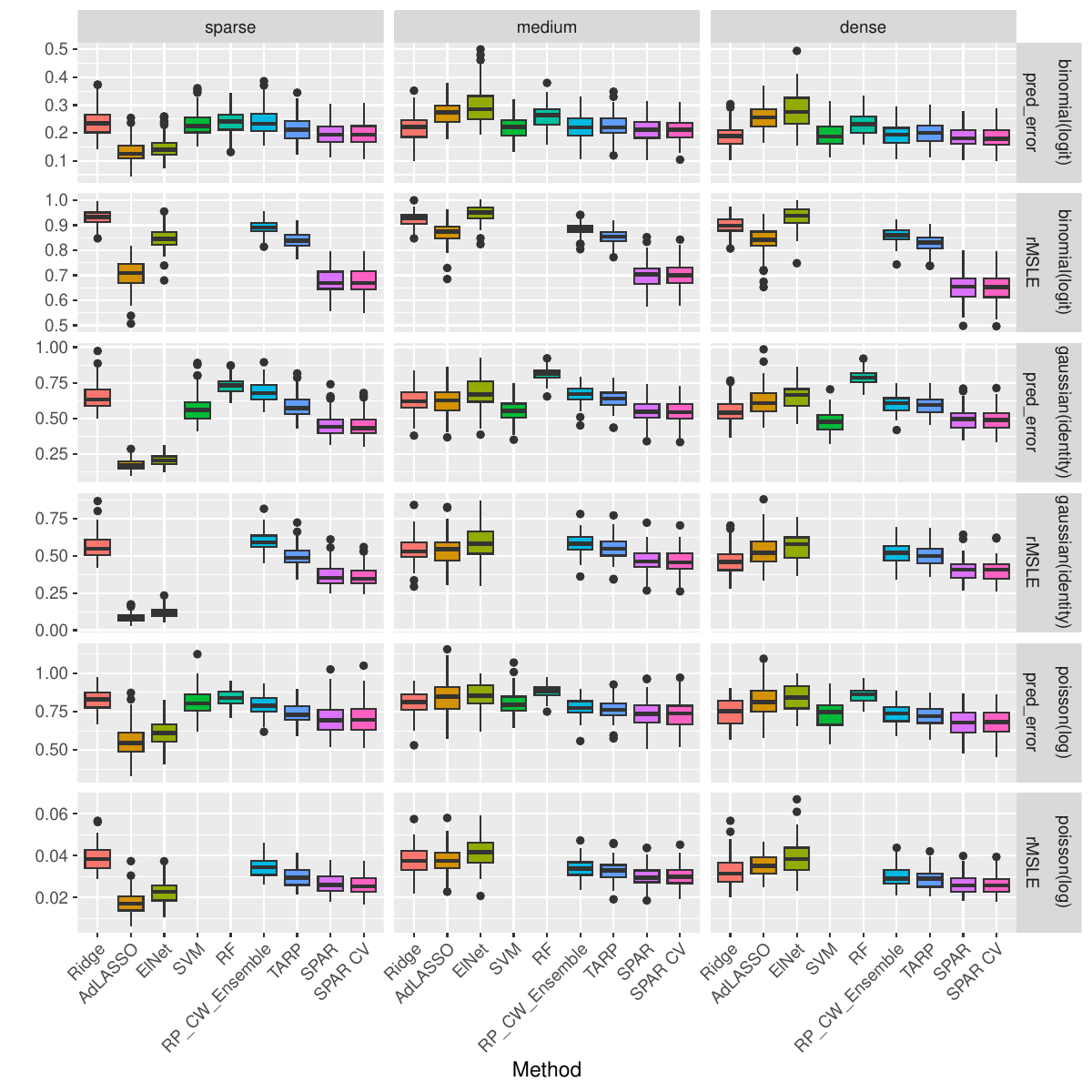}
\caption{Comparison of prediction performance ($1-$AUC for binomial, rMSPE otherwise) and link estimation (rMSLE, not reported for SVM and RF) over 100 repetitions ($n=200$, $p=2000$, medium sparsity and block $\boldsymbol{\Sigma}$).}
    \label{fig:pred}
  \end{figure}

\begin{figure}
\centering
\includegraphics[width=0.9\columnwidth,trim=0 20 0 6,clip]{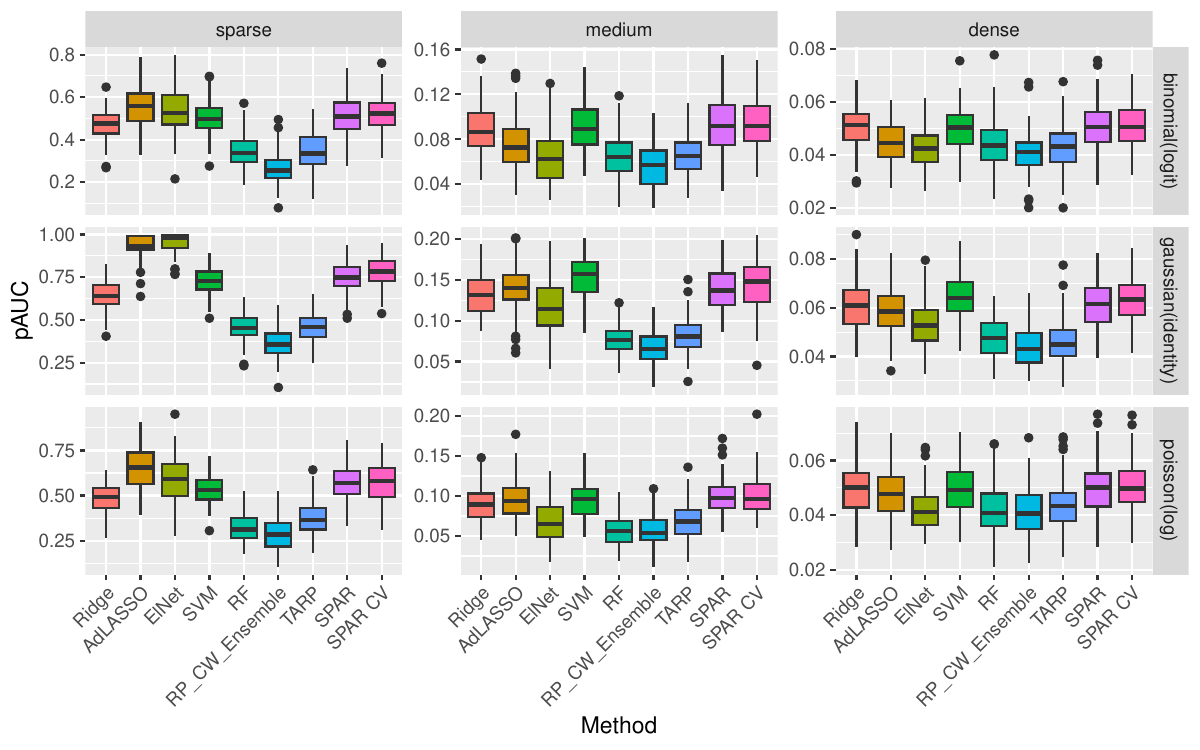}
\caption{Comparison of variable ranking (pAUC rescaled to [0, 1] for better comparison) over 100 repetitions for $n=200$, $p=2000$, medium sparsity and block $\boldsymbol{\Sigma}$.}
\label{fig:pAUC}
\end{figure}

To assess the performance across all investigated settings (including those in the supplementary materials), we rank the methods from best (1) to worst (11) for each replication and setting. Figure~\ref{fig:ranks} shows the average ranks (with 99\% confidence intervals) across all investigated settings. Aside from SVM in variable ranking, we find that SPAR-CV and SPAR achieve the best average ranks for all of the three measures. Using the Friedman and the post-hoc Nemenyi tests for multiple comparisons \citep{hollander2013nonparametric}, we can also report that SPAR and SPAR-CV are significantly better than all other methods for prediction and link estimation, and, together with SVM, they are also significantly better for variable ranking.
Even if SPAR and SPAR-CV are not best in every scenario, they perform well across-the-board, making them especially suitable  when the degree of sparsity is unknown. 

Finally, Figure~\ref{fig:sims_time} shows computing time for three increasing values of $p$ for the binomial family for the medium sparsity setting with block covariance. 
Most methods, except RF, SVM, and SIS, inherit  computational efficiency from \pkg{glmnet}, which uses a fast \textsf{C++} implementation for  canonical links.  SPAR is the second fastest for larger $p$ (after SIS) for the non-canonical family-link,  with  its time mostly spent fitting the $M=20$ marginal models and thus 
hardly affected by $p$.
We note that the ensemble of SparseCW random projections (RP\_CW\_Ensemble) is slower than SPAR due to projecting all $p$ variables, resulting in larger projection matrices. TARP is slower than SPAR likely due to using less sparse projections.
SPAR-CV, while slower than most methods for small $p$, scales efficiently with increasing $p$. 
The computing time for other family-link combinations follows  similar patterns based on whether the link is canonical, so those plots are omitted for brevity.

\begin{figure}
\centering
\includegraphics[width=0.8\columnwidth,trim=0 22 0 5,clip]{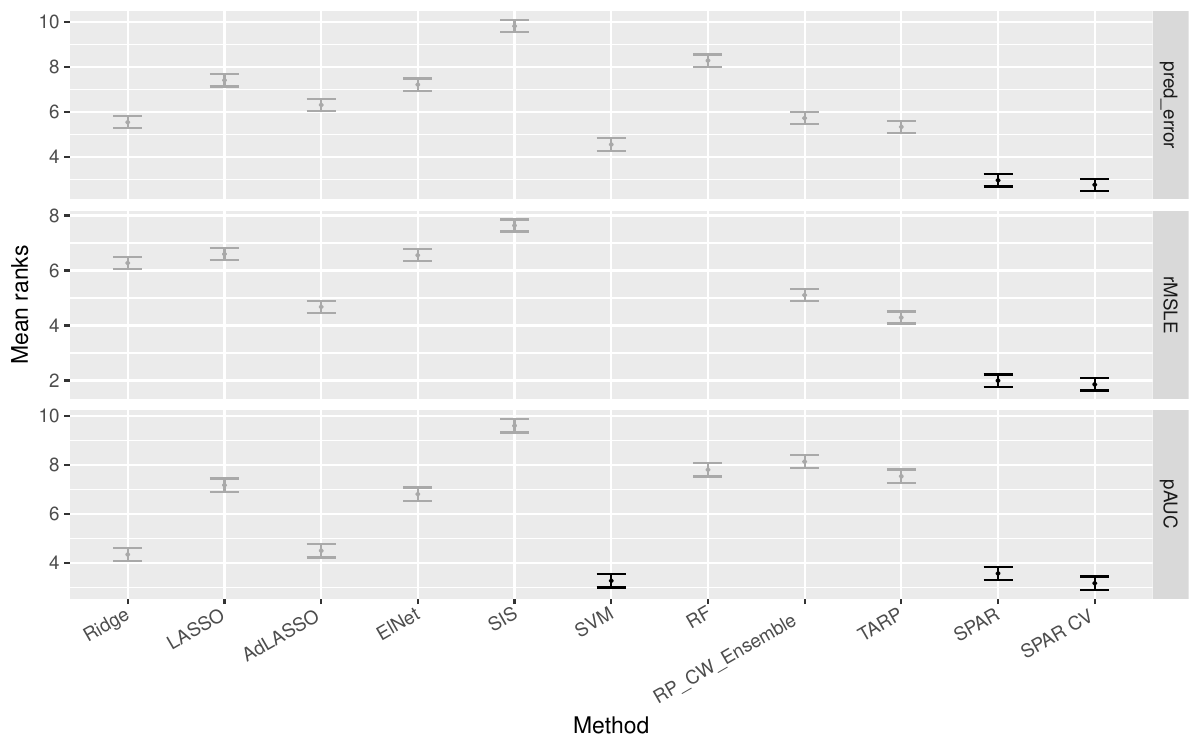}
\caption{Mean ranks with 99\% confidence intervals for prediction error, link estimation (rMSLE), and variable ranking (pAUC) across all investigated settings and $n_\text{rep}=100$ replications. Methods not significantly worse than the best method are colored black.}
\label{fig:ranks}
\end{figure}

\begin{figure}
\centering
\includegraphics[width=0.9\columnwidth,trim=0 5 0 5,clip]{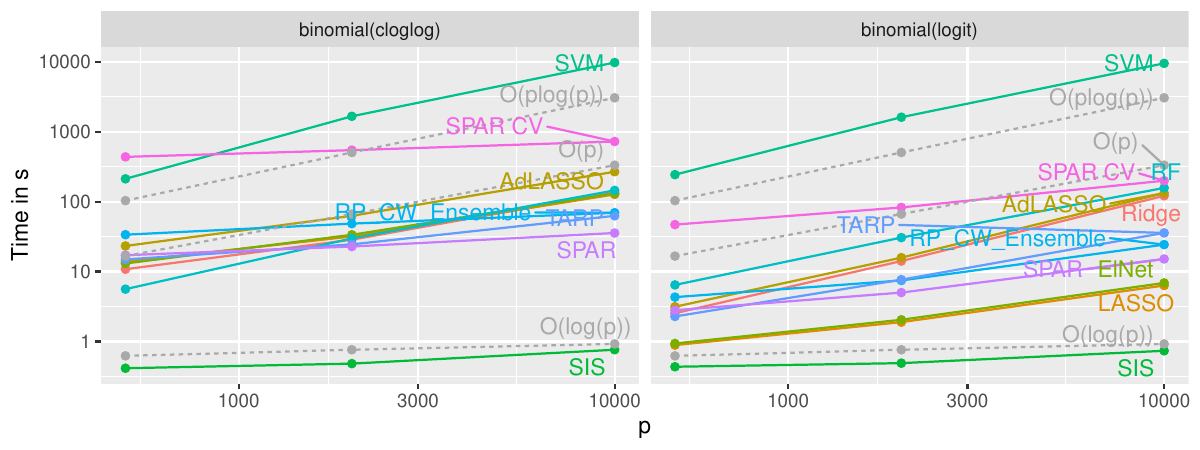}
\caption{Comparison of average computing time for increasing $p$, $n=200$, medium sparsity and block $\boldsymbol{\Sigma}$, for the binomial family.}
\label{fig:sims_time}
\end{figure}


\section{Data applications}\label{sec:data}

We illustrate the proposed method on one high-dimensional dataset with a count response and two datasets with binary response.
Table~\ref{tab:data_performance} shows the average prediction performance metrics for all analyzed datasets over $100$ random training/test splits of ratio three to one. For SPAR-CV, we use deviance as cross-validation measure for all datasets.
The other methods have been tuned as described in Section~\ref{sec:sims:SPAR}.

\begin{table}[t!]
  \centering
  \caption{Mean prediction metrics (with standard errors) on datasets over $100$ three-to-one training/test splits, the best three methods for each dataset and metric are marked in bold.}
  \label{tab:data_performance}
  \resizebox*{1.0\textwidth}{!}{
    \begin{tabular}{lrlrlrlrlrlrlrl}
\toprule
\multicolumn{1}{c}{ } & \multicolumn{2}{c}{FTIR spectra} & \multicolumn{4}{c}{Darwin} & \multicolumn{4}{c}{DLBCL} & \multicolumn{4}{c}{DLBCL\_extended} \\
\cmidrule(l{3pt}r{3pt}){2-3} \cmidrule(l{3pt}r{3pt}){4-7} \cmidrule(l{3pt}r{3pt}){8-11} \cmidrule(l{3pt}r{3pt}){12-15}
\multicolumn{1}{c}{Method } & \multicolumn{2}{c}{rMSPE} & \multicolumn{2}{c}{AUC} & \multicolumn{2}{c}{rMSPE} & \multicolumn{2}{c}{AUC} & \multicolumn{2}{c}{rMSPE} & \multicolumn{2}{c}{AUC} & \multicolumn{2}{c}{rMSPE} \\
\midrule
Ridge & \textbf{0.059} & (0.004) & 0.915 & (0.004) & 0.542 & (0.007) & \textbf{0.995} & (0.001) & 0.282 & (0.01) & 0.812 & (0.009) & 0.856 & (0.008)\\
LASSO & 0.079 & (0.006) & 0.914 & (0.004) & 0.509 & (0.01) & 0.963 & (0.005) & 0.456 & (0.02) & \textbf{0.942} & (0.007) & \textbf{0.601} & (0.016)\\
AdLASSO & 0.588 & (0.058) & 0.782 & (0.006) & 0.805 & (0.007) & 0.875 & (0.009) & 0.676 & (0.023) & 0.838 & (0.011) & 0.727 & (0.02)\\
ElNet & 0.078 & (0.006) & 0.921 & (0.003) & 0.491 & (0.008) & 0.979 & (0.004) & 0.346 & (0.017) & \textbf{0.964} & (0.005) & \textbf{0.514} & (0.017)\\
SIS & 0.424 & (0.051) & 0.898 & (0.004) & 0.600 & (0.007) & 0.902 & (0.009) & 0.678 & (0.021) & 0.902 & (0.009) & 0.678 & (0.021)\\
\addlinespace
SVM & 0.855 & (0.009) & \textbf{0.954} & (0.003) & \textbf{0.355} & (0.011) & 0.992 & (0.002) & \textbf{0.165} & (0.012) & 0.885 & (0.013) & 0.648 & (0.021)\\
RF & 0.126 & (0.009) & \textbf{0.957} & (0.003) & 0.444 & (0.006) & 0.960 & (0.004) & 0.544 & (0.01) & 0.912 & (0.007) & 0.679 & (0.011)\\
\addlinespace
RP\_CW\_Ensemble & 0.066 & (0.005) & 0.927 & (0.003) & 0.450 & (0.009) & 0.988 & (0.003) & 0.403 & (0.013) & 0.815 & (0.011) & 0.835 & (0.016)\\
TARP & 0.067 & (0.007) & 0.927 & (0.003) & 0.447 & (0.01) & 0.983 & (0.003) & 0.359 & (0.012) & \textbf{0.939} & (0.005) & \textbf{0.515} & (0.017)\\
SPAR & \textbf{0.036} & (0.004) & 0.935 & (0.003) & \textbf{0.439} & (0.013) & \textbf{0.994} & (0.002) & \textbf{0.179} & (0.013) & 0.921 & (0.007) & 0.627 & (0.022)\\
SPAR CV & \textbf{0.033} & (0.003) & \textbf{0.935} & (0.003) & \textbf{0.437} & (0.013) & \textbf{0.994} & (0.001) & \textbf{0.183} & (0.012) & 0.932 & (0.006) & 0.603 & (0.023)\\
\bottomrule
\end{tabular}
 }%
\end{table}

\subsection{FTIR spectra}
Fourier transform infrared (FTIR) spectroscopy is commonly used in tribology to analyze changes in oil samples during use. The dataset, introduced in \cite{PFEIFFER2022TribData}, consists of $n=34$ artificially altered automotive engine oil samples. Two types of artificial alteration were done in a laboratory by heating the oil and exposing it to dried air, simulating real-life degradation. For each of the $n=34$ samples, the difference FTIR spectra, i.e., the difference between the absorbances of the fresh oil and those of the degraded oils at $p=1814$ wavenumbers, and alteration duration (in hours) were recorded along with the type of alteration.

Table~\ref{tab:data_performance} shows the relative MSPE for various methods. SPAR(-CV) was estimated using both a (quasi) Poisson family and a Gaussian family, both with log links. We only present results for the Gaussian model, as it gave the best predictions. 
Moreover, SPAR(-CV) performed best overall, followed by ridge regression.

Figure~\ref{fig:trib_appl} (top panel) displays the difference spectrum for one sample, highlighting intervals with high or total absorption. These regions, typically non-informative due to hydrocarbon properties, are often pre-processed or discarded from analysis. The bottom panel shows the standardized coefficients estimated by SPAR-CV across $M=50$ marginal models, where the coefficients for each variable are sorted by their absolute values and displayed vertically as a color gradient.
Even without pre-processing, non-informative variables rarely appear in the models and have low coefficients when they do, demonstrating the reliability of the SPAR method. The distribution of coefficients indicates which wavenumbers correlate positively or negatively with longer alteration durations.

\begin{figure}[t!]
    \includegraphics[width=0.87\columnwidth,trim=5 20 0 0,clip]{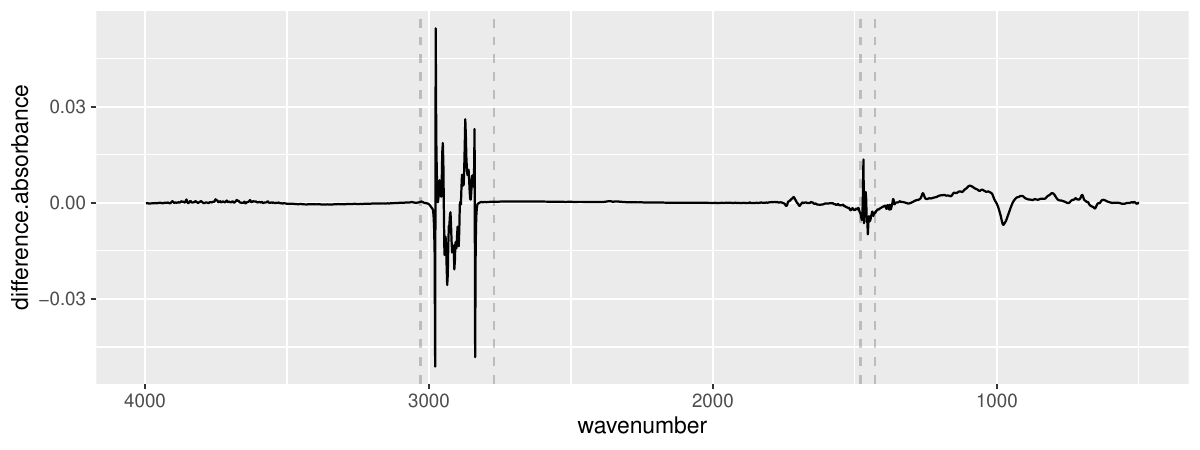}
    \includegraphics[width=1\columnwidth,trim=0 5 0 0,clip]{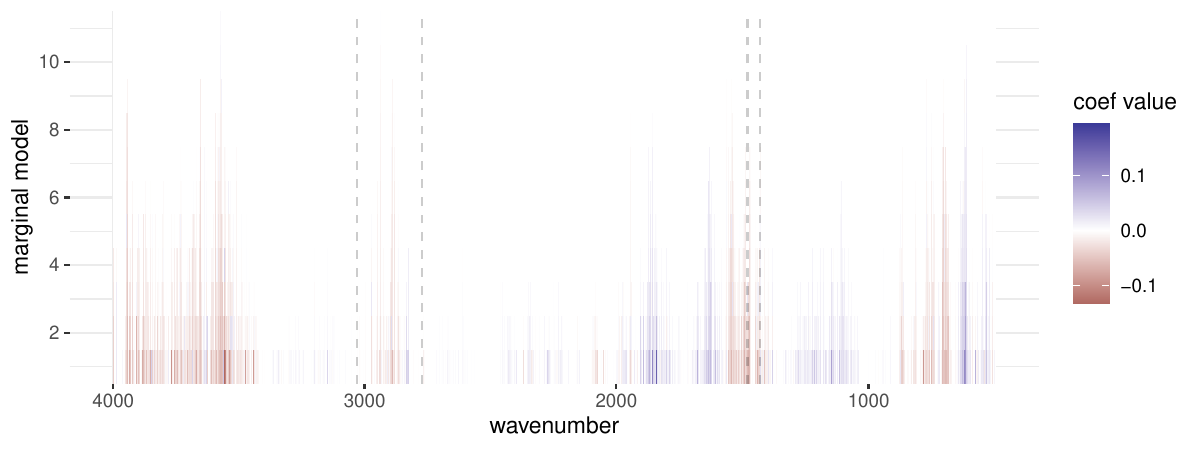}
  \caption{Difference FTIR spectrum for one oil sample in the tribology dataset (top) and coefficients of $p=1814$ wavenumbers estimated by SPAR-CV in each of $M=50$ marginal models (bottom). Marked intervals represent non-informative variables.}
  \label{fig:trib_appl}
  \end{figure}

\subsection{Darwin Alzheimer dataset} 
The dataset, introduced in \cite{CILIA2022darwin}, contains a binary response for Alzheimer's disease (AD) together with $p=450$ extracted variables from 25 handwriting tests (18 features per task) for 89 AD patients and 85 healthy people $(n=174)$ and can be downloaded from the  \href{https://archive.ics.uci.edu/dataset/732/darwin}{UC Irvine Machine Learning Repository}. 
The dataset has been reported to contain outliers, so before proceeding with the analysis, we screened 
for multivariate outliers and imputed them using the \emph{detect deviating cells} algorithm of \cite{rousseeuw2018detecting} implemented in \textsf{R} package \pkg{cellWise} \citep{cellwise_rpkg}. 
Table~\ref{tab:data_performance} presents the area under the ROC curve and the relative MSPE (Brier score)
as prediction metrics for the binary classification task.
For all methods based on GLMs (i.e., except SVMs and RFs), results are shown for the binomial family with logit link, as it provided the best performance among all investigated link functions. 
SPAR and SPAR-CV performed similarly and were outperformed in AUC by both SVM and RF, and in rMSPE by SVM alone. This suggests SPAR is a viable option for modeling this dataset, while offering lower computational cost.

Figure~\ref{fig:coefs_darwin} displays the estimated standardized coefficients for the $p=450$ variables grouped by feature, across all marginal models.
Feature blocks generally show either a positive or negative impact on the probability of AD across all 25 tasks (indicated by blocks of blue or red lines). 
For example, the probability of AD increases with the total time spent on a task
(total$\_$time). The number of pendowns (num$\_$of$\_$pendown, the number of times the pen hits the paper) is  positively associated with the AD likelihood for the first few tasks, but negatively correlated for the remaining tasks, with a strong negative association observed for the task ``Copy the fields of a postal order'', which requires  many pendowns when copying the individual fields. This may suggest that a lower number of pendowns in this task indicates failure to complete the task and thus a higher likelihood of AD.

\begin{figure}[t!]
    \centering
    \includegraphics[width=1.0\columnwidth,trim=0 5 0 10,clip]{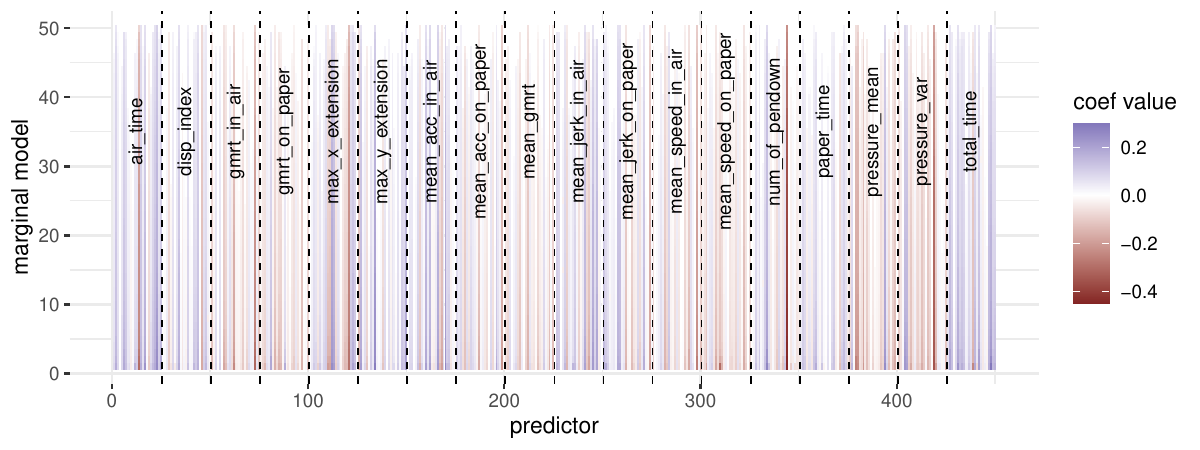}
    \caption{Estimated coefficients for the $p=450$ variables in the Darwin dataset, for each of the $M=50$ marginal models in the SPAR-CV algorithm.}
    \label{fig:coefs_darwin}
\end{figure}

\subsection{Diffuse large B-cell lymphoma (DLBCL)}
The \emph{Diffuse large B-cell lymphoma (DLBCL)} microarray dataset has been introduced in  
\cite{Shipp2002LymphData} and is available at the
\href{https://www.openml.org/search?type=data&status=active&id=45088}{OpenML platform}. It contains 
expression data for $p=5469$ genes of $n=77$ patients diagnosed with two different types of lymphomas: DLBCL ($58$ cases) and follicular lymphoma (FL, $19$ cases), making this a rather \emph{imbalanced} dataset.
Given that we do not know the true degree of sparsity in any of the employed datasets, we artificially extend the DLBCL dataset by randomly permuting the rows of the covariate matrix and appending the permuted covariates to the existing dataset twice. This results in an \emph{extended DLBCL dataset} with $n=77$ and $p=5469\times 3=16407$.

Again, Table~\ref{tab:data_performance} presents the area under the ROC curve and the relative MSPE (Brier score)
as prediction metrics and for all methods based on GLMs results are shown for the binomial family with logit.
SPAR and SPAR-CV perform very well on the original dataset, being outperformed only by SVM in terms of rMSPE and by ridge in terms of AUC. 
For the extended dataset, the best-performing methods in terms of AUC are elastic net followed by LASSO and TARP.
SPAR(-CV) performs satisfactorily, being outperformed only by elastic net and TARP in terms of rMSPE.
The superior performance of the sparse methods is not unexpected and is in line with the simulation results, as the degree of sparsity in this dataset has been artificially increased through the permutation of the variables.

\section{Conclusion}\label{sec:conclusion}

In this paper, we propose a novel data-driven random projection method to be employed in high-dimensional GLMs, which efficiently reduces the dimensionality of the problem while preserving essential information between the response and the (possibly correlated) predictors.
We achieve this by using ridge-type estimates of the regression coefficients to construct the random projection matrix which should accurately recover the true regression coefficients. These coefficients can also be employed for  variable screening, which can be used prior to random projection to further reduce the dimensionality of the problem.

A critical aspect of the proposed method is the selection of the penalty term for the ridge-type estimator. 
The penalty should generally be small to avoid over-regularization. 
However, determining the optimal size of the penalty has proven to be a non-trivial task. 
For linear regression, a ridge estimator with penalty converging to zero has shown good properties \citep{Wang2015HOLP,parzer2024sparse}. In this paper,
we derive the analytical formula for such an estimator in GLMs with canonical links and find that this estimator leads to lower predictive performance for non-Gaussian families, likely due to overfitting. More generally, there is no \emph{one-size-fits-all} penalty value for all families.
Instead, we advocate for a data-driven approach by decreasing the penalty value as long as the improvement in a goodness-of-fit criterion (e.g., deviance) exceeds a certain threshold (e.g., 0.8 for non-Gaussian families and $0.999$ for Gaussian) seems to be the best strategy.

Through extensive simulations, we show that integrating  multiple probabilistic variable screening and projection steps into an ensemble of medium-sized GLMs improves prediction accuracy and variable ranking, without too much computational cost. To implement this method, we adapt the SPAR algorithm from \cite{parzer2024sparse}, ensuring that it is tailored to the specific requirements of high-dimensional GLMs
and show that 
the method achieves top overall performance aggregated over all investigated scenarios, making it a valid choice when the true degree of sparsity is not known in practice. At the cost of higher computation time, which still scales well with $p$, the method can, to some degree, benefit from cross-validation, most notably in terms of ranking the variables based on their relevance.

Finally, three real data sets illustrate the interpretability and superior prediction performance of the proposed approach.
A potential extension  includes adapting the method to multivariate GLMs (e.g., multinomial) and multivariate responses (e.g.,~multivariate linear regression). A key extension in this direction would be designing a data-driven RP that can preserve the multivariate structure in the data while also being straightforward and fast to compute. Additionally, ways of incorporating non-linearities in the RP could be explored.

\section*{Supplementary materials}

The code to produce all simulations and data applications can be found in \url{https://github.com/RomanParzer/SPAR_GLM_Paper_Code}, and the R-package performing the SPAR-algorithm can be found on \url{https://github.com/RomanParzer/SPAR}.


\section*{Declaration of conflicting interests}
The authors declared no potential conflicts of interest with respect to the research, authorship and/or
publication of this article.

\section*{Funding}
Roman Parzer and Laura Vana-Gür acknowledge funding from the Austrian Science Fund (FWF) for
the project ``High-dimensional statistical learning: New methods to advance
economic and sustainability policies'' (ZK 35).

\begin{appendices}

\section{Proof of Theorem~\ref{theorem:lim_beta_lam}}\label{sec:app:proof}
To prove Theorem~\ref{theorem:lim_beta_lam}, we will need the following Lemma.
\begin{lemma}\label{lemma:betalam}
  Let $f_\lambda(\boldsymbol{\beta})= -\tilde \ell( \boldsymbol{\beta}) + \frac{\lambda}{2}\sum_{j=1}^p{\boldsymbol{\beta}}_j^2$ and 
  $\boldsymbol{\beta}_\lambda :=\text{argmin}_{{\boldsymbol{\beta}}\in\mathbb{R}^p}f_\lambda(\boldsymbol{\beta})$.\linebreak
  Then, $\lambda\boldsymbol{\beta}_\lambda \to \boldsymbol{0}$ for $\lambda\to 0$.
\end{lemma}

\begin{proof}[Proof of Lemma~\ref{lemma:betalam}]
 Each summand $f(\theta_i) = y_i\theta_i - b(\theta_i)$ of $\tilde\ell(.)$ is a concave function of $\theta_i$ and has a unique maximum at $\hat \theta_i = (b')^{-1}(y_i)$ if $(b')^{-1}(y_i)\in\mathbb{R}$, or is unbounded otherwise (as is the case for a Bernoulli variable with logit-link, where $(b')^{-1}(y_i) = \log(y_i/(1-y_i))$ and $y_i\in\{0,1\}$).
  Note that for the family with the canonical link and defined values $g(y_i)$, we have $\theta(\boldsymbol{\beta},\boldsymbol{x}_i) = \boldsymbol{x}_i'\boldsymbol{\beta}$ and, therefore, $f_\lambda$ is strictly convex and has a unique minimum.
  We will even show the stronger result, that $\lambda\|\boldsymbol{\beta}_\lambda\|^2 \to 0$. If that were not the case, there exists some $\varepsilon>0$ such that $\forall \Lambda \exists \lambda < \Lambda$ such that $\lambda\|\boldsymbol{\beta}_\lambda\|^2>\varepsilon$.
  Let $\Lambda:=\varepsilon/\|\boldsymbol{\hat\beta}\|^2$, where $\boldsymbol{\hat\beta} = \boldsymbol{X}'(\boldsymbol{X}\boldsymbol{X}')^{-1}g(\boldsymbol{y})$.
  According to Section~\ref{sec:model}, any $\boldsymbol{\beta}$ maximizes $\tilde\ell$, 
  if and only if $\boldsymbol{x}_i'\boldsymbol{\beta} = g(y_i)$ for all $i\in[n]$ or 
  $\boldsymbol{X}\boldsymbol{\beta}=g(\boldsymbol{y})$ in matrix notation, since we use the canonical link $g=(b')^{-1}$.
It is easy to see that $\boldsymbol{\hat\beta}$ satisfies this condition, and has the smallest $L_2$-norm among all such $\boldsymbol{\beta}$ (see e.g. proof of Lemma 3 in \cite{parzer2024sparse}).
  For the $\lambda$ corresponding to $\Lambda$, we then get
  \begin{align*}
    f_\lambda(\boldsymbol{\hat\beta})-f_\lambda(\boldsymbol{\beta}_\lambda) &= \underbrace{\tilde\ell(\boldsymbol{\beta}_\lambda) - \tilde\ell(\boldsymbol{\hat\beta})}_{\leq 0} + \underbrace{\frac{\lambda}{2}}_{< \Lambda/2} \|\boldsymbol{\hat\beta}\|^2 - \underbrace{\frac{\lambda}{2} \|\boldsymbol{\beta}_\lambda\|^2}_{> \varepsilon/2} \\
    &< \frac{\Lambda}{2}\|\boldsymbol{\hat\beta}\|^2 - \frac{\varepsilon}{2} = 0,
  \end{align*}
  which contradicts the definition of $\boldsymbol{\beta}_\lambda$.
\end{proof}

\begin{proof}[Proof of Theorem~\ref{theorem:lim_beta_lam}]
  We need to show that $\boldsymbol{\beta}_\lambda \to \boldsymbol{\hat\beta}$  for $\lambda \to 0$, where $\boldsymbol{\hat\beta} = \boldsymbol{X}'(\boldsymbol{X}\boldsymbol{X}')^{-1}g(\boldsymbol{y})$.
  Let $\boldsymbol{P\beta}_\lambda = \boldsymbol{\hat \beta} + (\boldsymbol{I}-\boldsymbol{X}'(\boldsymbol{X}\boldsymbol{X}')^{-1}\boldsymbol{X})\boldsymbol{\beta}_\lambda$ 
  be the orthogonal projection of $\boldsymbol{\beta}_\lambda$ to the (affine) subspace $\{\boldsymbol{X}\boldsymbol{\beta}=g(\boldsymbol{y})\}= \boldsymbol{\hat\beta} + \text{ker}(\boldsymbol{X})$. 
   Then
  \begin{align*}
    \|\boldsymbol{\hat\beta} - \boldsymbol{\beta}_\lambda\|^2 &= \|\boldsymbol{\hat \beta} - \boldsymbol{P\beta}_\lambda\|^2 + \|\boldsymbol{P\beta}_\lambda - \boldsymbol{\beta}_\lambda\|^2 = \\
    &= \|(\boldsymbol{I}-\boldsymbol{X}'(\boldsymbol{X}\boldsymbol{X}')^{-1}\boldsymbol{X})\boldsymbol{\beta}_\lambda\|^2 + 
    \|\underbrace{\boldsymbol{\hat\beta} - \boldsymbol{X}'(\boldsymbol{X}\boldsymbol{X}')^{-1}\boldsymbol{X}\boldsymbol{\beta}_\lambda}_{=  \boldsymbol{X}'(\boldsymbol{X}\boldsymbol{X}')^{-1}(g(\boldsymbol{y}) - \boldsymbol{X}\boldsymbol{\beta}_\lambda)}\|^2.
  \end{align*}
  For the canonical link, $f_\lambda$ is explicitly given by 
  $f_\lambda(\boldsymbol{\beta}) = \sum_{i=1}^n b(\boldsymbol{x}_i'\boldsymbol{\beta} ) - y_i\boldsymbol{x}_i'\boldsymbol{\beta} + \lambda \|\boldsymbol{\beta}\|^2/2$
  with the following first order optimality conditions at $\boldsymbol{\beta}_\lambda$.
  \begin{align*}
    0 &= \frac{\partial f_\lambda}{\partial \boldsymbol{\beta}}(\boldsymbol{\beta}_\lambda) = \boldsymbol{X}'(b'(\boldsymbol{X}\boldsymbol{\beta}_\lambda) - \boldsymbol{y}) + \lambda \boldsymbol{\beta}_\lambda.
  \end{align*}
This implies, that $\boldsymbol{\beta}_\lambda = -\frac{1}{\lambda}\boldsymbol{X}'(b'(\boldsymbol{X}\boldsymbol{\beta}_\lambda) - \boldsymbol{y}) \in \text{span}(\boldsymbol{X}')$ and, therefore,
$(\boldsymbol{I}-\boldsymbol{X}'(\boldsymbol{X}\boldsymbol{X}')^{-1}\boldsymbol{X})\boldsymbol{\beta}_\lambda = 0$.
  Rewriting the optimality conditions and using Lemma~\ref{lemma:betalam} also show that
  \begin{align*}
    \boldsymbol{X}'(b'(\boldsymbol{X}\boldsymbol{\beta}_\lambda) - \boldsymbol{y}) = - \lambda \boldsymbol{\beta}_\lambda \to \boldsymbol{0} \text{ for } \lambda \to 0.
  \end{align*}
  Since $\text{rank}(\boldsymbol{X})=n$, this is only possible if $b'(\boldsymbol{X}\boldsymbol{\beta}_\lambda) \to \boldsymbol{y}$ for $\lambda\to 0$, and, since the canonical link $g$ is continuous, this implies 
  $g(b'(\boldsymbol{X}\boldsymbol{\beta}_\lambda)) =  \boldsymbol{X}\boldsymbol{\beta}_\lambda\to g(\boldsymbol{y})$.
  Therefore, the second term also vanishes for $\lambda \to 0$ and  $\boldsymbol{\beta}_\lambda \to \boldsymbol{\hat\beta}$.
\end{proof}

\section{Additional tables and figures for simulations}

\begin{table}[htbp!]
    \centering
    \resizebox{1.0\textwidth}{!}{
\begin{tabular}{lllllllll}
\toprule
\multicolumn{1}{c}{Method} & \multicolumn{2}{c}{binomial(logit)} & \multicolumn{2}{c}{gaussian(identity)} & \multicolumn{2}{c}{gaussian(log)} & \multicolumn{2}{c}{poisson(log)} \\
\midrule
L2\_cv & 9.993 & (1.399) & 28.435 & (0.839) & 1458.825 & (140.805) & 580.668 & (45.346)\\
L2\_dev08 & 0.677 & (0.008) & 10.065 & (0.182) & 1448.049 & (32.237) & 77.657 & (1.487)\\
L2\_dev095 & 0.094 & (0.001) & 2.695 & (0.038) & 345.963 & (6.554) & 16.759 & (0.305)\\
L2\_dev0999 & 0.001 & (0) & 0.246 & (0.002) & 23.346 & (0.316) & 0.791 & (0.042)\\
\bottomrule
\end{tabular}
}
\caption{Average chosen $\lambda$ for different estimators over $100$ replications ($n=200,p=2000$, 
  block-diagonal $\boldsymbol{\Sigma}$ and medium sparsity).}
    \label{tab:scr_lambda}
\end{table}

\begin{figure}[t!]
  \centering
  \includegraphics[width=1.0\columnwidth]{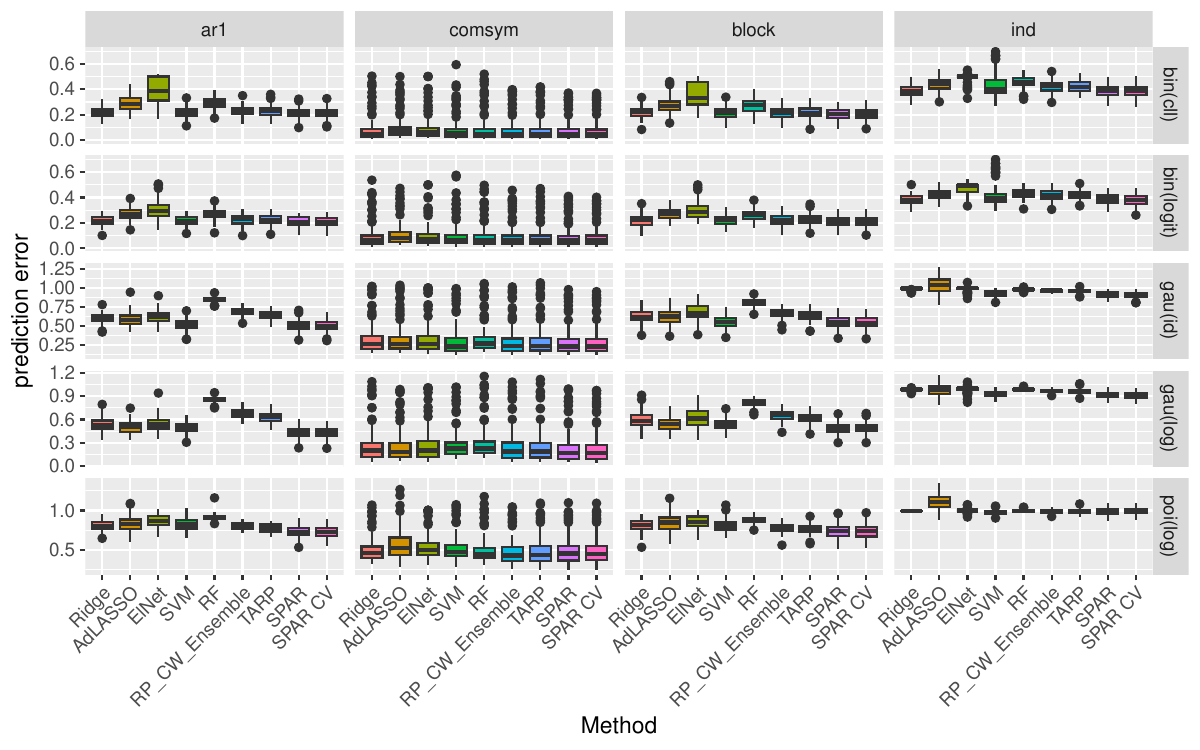}
  \caption{Prediction error for different covariance structures with medium setting for $p=2000$}
  \label{fig:sim_pred_cov_settings}
\end{figure}

\begin{figure}[t!]
  \centering
  \includegraphics[width=1.0\columnwidth,trim=0 20 0 0,clip]{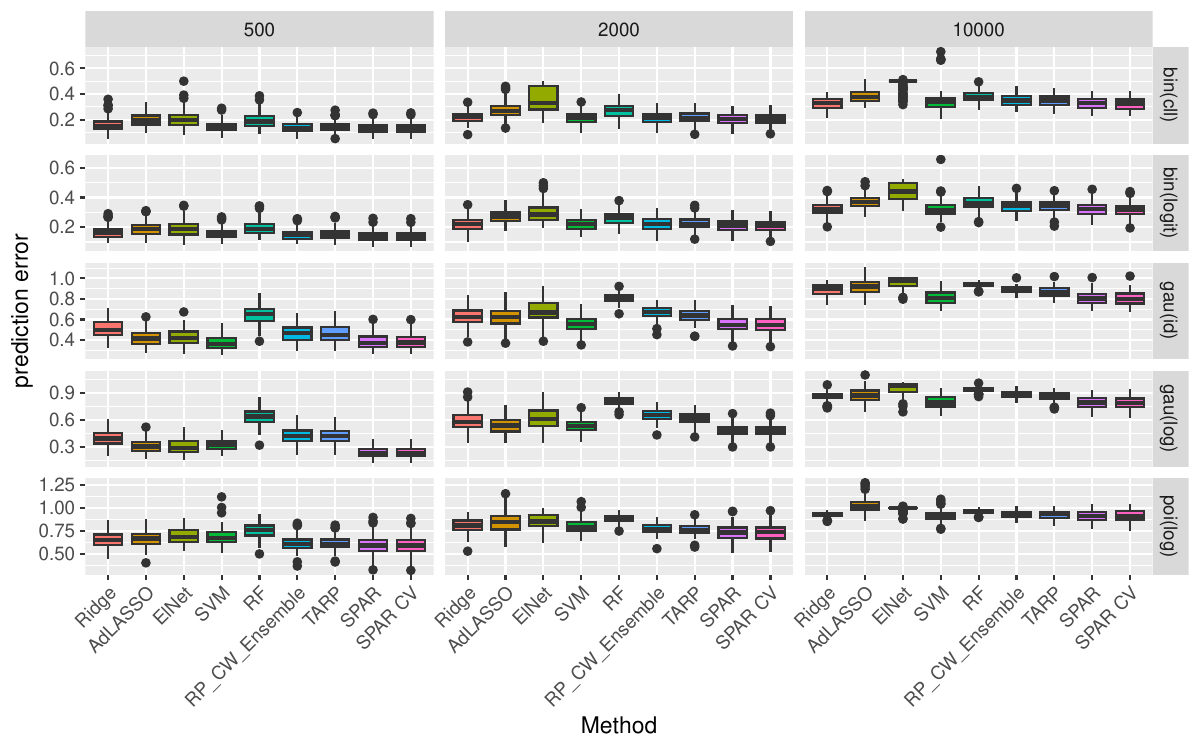}
  \caption{Prediction error for increasing $p$ with medium setting and block-diagonal covariance}
  \label{fig:sim_pred_p}
\end{figure}

\end{appendices}

\bibliographystyle{unsrtnat}


\end{document}